\documentclass[runningheads,a4paper]{llncs}

\usepackage{fullpage}
\usepackage{float}

\usepackage{amssymb}								
\usepackage{tikz} \usetikzlibrary{chains,positioning,scopes,arrows,plotmarks}	
\usepackage{comment}
\usepackage{array}								
\usepackage{amsmath,mathtools}							
\usepackage{tabto}								
\usepackage{xspace}								

\newcommand{\bac}{\backslash}
\newcommand{\seb} {\subseteq}
\newcommand{\matN}{\mathbb N}
\newcommand{\matR}{\mathbb R}

\newcommand{\cL}{{\cal L}}
\newcommand{\cM}{{\cal M}}

\newcommand{\cP}{{\cal P}}
\newcommand{\cS}{{\cal S}}
\newcommand{\cW}{{\cal W}}

\newcommand{\tF}{{\tt F}}
\newcommand{\tI}{{\tt I}}
\newcommand{\tL}{{\tt L}}
\newcommand{\tM}{{\tt M}}
\newcommand{\SAT}{\mbox{{\sc Sat}}\xspace}
\newcommand{\SATbar}{\overline{\mbox{{\sc Sat}}}\xspace}
\newcommand{\RAE}{\mbox{{\sc Rae}}\xspace}
\newcommand{\BZ}{\mbox{{\sc Bz}}\xspace}

\begin{document}

\mainmatter

\title{Satisfaction in societies with opinion leaders and mediators:\\properties and an axiomatization}
\author{Fabi\'an Riquelme\thanks{This author is funded by grant BecasChile of the ``National Commission for Scientific and Technological Research of Chile'' (CONICYT), and also supported by 2009SGR-1137 (ALBCOM).}}

\institute{Departament de Llenguatges i Sistemes Inform\`atics, Universitat Polit\`ecnica de Catalunya,\\
              Campus Nord-Ed. Omega, Jordi Girona 1-3, 08034 Barcelona, Spain.\\
	\email{farisori@lsi.upc.edu}
}

\date{}

\maketitle

\begin{abstract}
In this paper we propose the opinion leader-follower through mediators systems ---OLFM systems--- a multiple-action collective choice model for societies. In those societies three kind of actors are considered: {\em opinion leaders} that can exert certain influence over the decision of other actors, {\em followers} that can be convinced to modify their original decisions, and {\em independent actors} that neither are influenced nor can influence; {\em mediators} are actors that both are influenced and influence other actors.
This is a generalization of the opinion leader-follower systems ---OLF systems--- proposed by~\cite{BRS11}.

The {\em satisfaction} score is defined on the set of actors. For each actor it measures the number of society initial decisions in which the final collective decision coincides with the one that the actor initially selected. We generalize in OLFM systems some properties that the satisfaction score meets for OLF systems. By using these properties, we provide an axiomatization of the satisfaction score for the case in which followers maintain their own initial decisions unless all their opinion leaders share an opposite inclination. This new axiomatization generalizes the one given by~\cite{BRS12} for OLF systems under the same restrictions.

\end{abstract}

\section{Introduction}
\label{sec:intro}

Opinion leadership is a well known and established model for communication policy in sociology and marketing. It comes from the {\em two-step flow of communication} theory proposed since the 1940s~\cite{LBG68}. This theory recognizes the existence of collective decision making situations in societies formed by actors called {\em opinion leaders}, who exert influence over other kind of actors called the {\em followers}, becoming in a two-step decision process~\cite{LBG68,KL55}. In the first step of the process, all actors receive information from the environment, generating their own decisions; in the second step, a flow of influence from some actors over others is able to change the choices of some of them~\cite{Tro66}.

In general, a {\em collective (choice) decision making model} for a finite set of actors defines a {\em collective (choice) decision function}. For any set of decisions taken independently by the actors, and represented by an initial decision vector or an initial choice vector, this function assigns a collective decision, i.e., an outcome that corresponds to one of the values ``yes'' or ``no'', ``pass'' or ``reject'', ``agree'' or ``disagree'', ``true'' or ``false'', $1$ or $0$, etc. As usual in decision theory, we focus on a binary set of possible decisions that the actors can take.

Motivated by the theoretical study of the effects that different opinion leader-follower structures can exert in collective decision making systems, a {\em satisfaction} score was defined in~\cite{BRS11,BRS12} for an {\em opinion leader-follower collective decision system} ---OLF system, in short--- which represents societies with opinion leaders, followers and {\em independent actors}. The actors of this latter type neither are influenced nor can influence other members of the society.

According to~\cite{BRS11}, the {\em satisfaction} of an actor in a society refers to the number of possible decisions that all actors can take as a group, such that the collective decision coincides with the decision taken by the actor in the initial choice vector. This general formulation allow us to define the satisfaction score for a generic collective decision making model. In OLF systems, satisfaction is an essential notion, since it is the most extensive notion to characterize the position of an actor in a society~\cite{BRS12}.

There are several properties that the satisfaction score meets specifically for OLF systems~\cite{BRS11}. These properties are related to the variation of satisfaction for certain actors when the relationship with their predecessors or successors is modified. If we restrict our attention to {\em unanimity}, i.e., the specific case when the followers maintain their own initial decisions unless all their opinion leaders share an opposite inclination ---in which case the follower is convinced by their opinion leaders--- then it can be defined an axiomatization of the satisfaction score~\cite{BRS12}. Unanimity is a usual restriction, whose application has been experimentally used in multiple times. For instance, in a recent experiment it was found that actors were more likely to conform the attitudes expressed by a unanimous group than by a non-unanimous group~\cite{VL09}.

Several efforts have been made to find axiomatizations for measures or scores related to decision systems and cooperative games. Moreover, several of these efforts concern to the Banzhaf value~\cite{DS79,Leh88,Bri97,Bri10}, which is a well known power index of simple games, that we shall see is closely related with the satisfaction. The axiomatizations are relevant to know exactly what properties satisfies a measure in a collective decision making model.

In this paper we propose a generalization of the OLF systems, namely the OLFM systems, based on a collective decision making model introduced in~\cite{MRS12}, that provides actors called {\em mediators} that behave at the same time like opinion leaders and followers. The mediators act as intermediate layers of actors, that moderate the influence of opinion leaders over followers, helping us to describe a ``more-than-two-step'' flow of communication. This scenario allows to analyze societies where there are several layers of influence, establishing a more complex hierarchy among the different actors. Furthermore, we also show an axiomatization for the satisfaction score in OLFM systems, by using generalized versions of the properties used for the same score in OLF systems.

The paper is organized as follows.
In Section~\ref{sec:SAT} we define satisfaction and associate this measure with two well known power indices of simple game theory, namely the Rae index and the Banzhaf index.
Section~\ref{sec:OLF} presents the original OLF systems defined in~\cite{BRS11,BRS12}, showing the collective decision function of the model, and the axiomatization of the satisfaction for this model.
Section~\ref{sec:OLFM} presents the generalization of the OLF model through the use of mediators, it defines the generalized versions of the properties of satisfaction, and also defines the new axiomatization.
The paper finishes with some conclusions and remarks.

\section{Simple games and Satisfaction}
\label{sec:SAT}

Let $V$ be a set of actors, as usual $n=|V|$ denotes the number of actors of the system. In general, a {\em collective decision making model} $\cM$ for a set of $n$ actors defines a {\em collective decision function} $C_\cM(x)$, where $x\in\{0,1\}^n$ is the initial decision vector of the actors, assigning one of the values $1$ or $0$ as collective decision; value $1$ represents the decision ``yes'', and value $0$ the decision ``no'' of an actor.

Abusing of notation, we may consider a collective decision making model $C_\cM(X)$ instead of $C_\cM(x)$, where $i\in X\seb 2^V$ if and only if the $i$-th component of $x$ is $1$, i.e., $x_i=1$.

\begin{definition}
We said that $x\in\{0,1\}^n$ is an {\em initial decision vector}, when $x_i$ represents the initial decision of the $i$-th actor of some decision system.
\end{definition}

A well known and deeply studied decision system is the one of simple games, a class of cooperative games equivalent to monotone Boolean functions~\cite{TZ99}.

\begin{definition}
A {\em simple game} is a pair $\Gamma=(V,\cW)$, where $V$ is a set of actors or players, and $\cW$ its set of {\em winning coalitions} such that for all $X,Y\seb V$, if $X\in\cW$ and $X\seb Z$, then $Z\in\cW$.
\end{definition}

Observe that we can associate to any simple game a collective decision function in a natural way. Let $x\in\{0,1\}^n$ be an initial decision vector of the players, the collective decision function associated to $\Gamma$ is defined as follows:

\[
C_\Gamma(x)=
\begin{cases}
 1 & \mbox{if } X(x)\in\cW,\\
 0 & \mbox{otherwise}
\end{cases}
\]
where $X(x)=\{i\in V\mid x_i=1\}$. In simple games as a collective decision making model, it is relevant to study the importance of the players in the decision-making process. In this context, the measures or scores of the players are known as {\em values} or {\em power indices}. One of the most classic and popular power indices is the {\em Banzhaf value}, that corresponds to the proportion of coalitions in which a player plays a critical role~\cite{Pen46,Ban65,Col71}.

\begin{definition}
Let $\Gamma=(V,\cW)$ be a simple game and $i\in V$ an actor or player. The {\em Banzhaf value} $\BZ(i)$ is the number of coalitions in which $i$ is {\em critical}, i.e., $\BZ(i)=|\{X\seb V\mid X\in\cW\mbox{ and }X\setminus\{i\}\notin\cW\}|$.
\end{definition}

Besides the critical players, now we mention other two kind of players~\cite{TZ99}.

\begin{definition}\label{def:dummy_dictator}
Let $(V,\cW)$ be a simple game and $i\in V$ a player:
 \begin{itemize}
  \item $i$ is a {\em dummy} if and only for all $X\seb V$, $X\in\cW$ implies $X\setminus\{i\}\in\cW$;
  \item $i$ is a {\em dictator} if and only if for all $X\seb V$, $X\in\cW$ if and only if $i\in X$.
 \end{itemize}
\end{definition}

Now we define the satisfaction score in general terms, for any collective decision making model.

\begin{definition}\label{def:SAT}
Let $\cM$ be a collective decision making model over a set of $n$ actors.
The {\em satisfaction score} of the actor $i$ is defined as follows:
$$\SAT_\cM(i) = |\{x\in\{0,1\}^n\mid C_\cM(x)=x_i\}|.$$
\end{definition}

It is interesting to note that when the collective decision making model $\cM$ is monotonic, with respect to inclusion, the satisfaction score coincides with the known {\em Rae index}.
This power index was introduced by~\cite{Rae69} for anonymous games and afterwards it was applied by~\cite{DS79} for simple games, being defined as follows:
$$\RAE(i)=|\{X\seb V\mid i\in X\in\cW\mbox{ or } i\notin X\notin\cW\}|.$$

In the context of simple games,~\cite{DS79} established an affine-linear relation between the Rae index and the Banzhaf value~\cite{LMF06}:
\begin{equation}
\label{eq:RAE_BZ}
\SAT(i)=\RAE(i)=2^{n-1}+\BZ(i)
\end{equation}

It is clear that this equality holds for any collective decision making model that is monotonic. Later, in Lemma~\ref{lem:monotony} we show that the OLF systems ---and therefore its generalization, the OLFM systems--- are monotonic, so the just mentioned equality also applies for them.

\begin{lemma}\label{lem:Rae_dummy_dictator}
Let $\cM$ be any monotonic collective decision making model. For any player $i$ we have $\SAT(i)\geq 2^{n-1}$. Moreover, if $i$ is a dummy then $\SAT(i)=2^{n-1}$, and if $i$ is a dictator then $\SAT(i)=2^n$.
\end{lemma}

\begin{proof}
The sentence $\SAT(i)\geq 2^{n-1}$ is deduced from the equation~(\ref{eq:RAE_BZ}). It is well known that if $i$ is a dummy, then $\BZ(i)=0$, so then $\SAT(i)=2^{n-1}$. If $i$ is dictator, then for any coalition $X\seb V$, if $X\in\cW$ then $i\in X$, and if $X\not\in\cW$ then $i\not\in X$, so hence $\SAT(i)=2^n$. \qed
\end{proof}

Definition~\ref{def:SAT} helps to see the relationship among the satisfaction and both the Rae index and the Banzhaf value. However, there exists an equivalent definition of satisfaction, that facilitates the proofs of the results of the paper~\cite{BRS11}.

\begin{definition}
Let $\cM$ be a collective decision making model over a set of $n$ actors.
The {\em satisfaction score} of the actor $i$ can also be defined as follows:

\begin{equation}\label{SATwithSATbar}
\SAT_\cM(i) = \sum_{x\in\{0,1\}^n}\SATbar_\cM(i,x)
\end{equation}
where
\[
\SATbar_\cM(i,x)=
\begin{cases}

 1 & \mbox{if } C_\cM(x)=x_i\\
 0 & \mbox{otherwise}
\end{cases}
\]
\end{definition}

In what follows, we use simply $\SATbar(i,x)$ and $\SAT(i)$ when there is no risk of confusion about $\cM$.

\section{OLF systems}
\label{sec:OLF}

The outcome of the model presented in this section is computed by influence interactions in a directed bipartite graph. At the end of the process all actors arrive to an stable solution and the collective decision function corresponds to the {\em simple majority voting system}.

We use standard notation for graphs~\cite{Bol98}: $G=(V,E)$ is a directed graph, $V(G)$ denotes the set of vertices or actors, and $E(G)$ is the set of edges, i.e., the set of relations among the actors. We use simply $V$ and $E$ when there is no risk of confusion.
For each $i\in V$, $S_G(i)=\{j\in V\mid (i,j)\in E\}$ denotes the set of {\em successors} of $i$, and $P_G(i)=\{j\in V\mid (j,i)\in E\}$ the set of {\em predecessors} of $i$.
Let $\delta^-(i)=|P_G(i)|$ and $\delta^+(i)=|S_G(i)|$ be the indegree and the outdegree of the node $i$, respectively.

Furthermore, we denote $\cP(V(G))$ as the power set of $V(G)$. Let $x\in\{0,1\}^n$ be a binary vector and $x_i\in\{0,1\}$ its $i$-th component, then $x+i$ and $x-i$ denote an addition and a substraction of the $i$-th actor of $x$, respectively, that is, to assign $1$ or $0$ to the $i$-th component of $x$, respectively.

The model can be formalized as follows.

\begin{definition}\label{def:OLF}
An {\em opinion leader-follower system} $\cS$ ---{\em OLF system}, in short--- for a set of $n$ actors is given by a bipartite digraph $G=(V,E)$, representing the actors' relation, such that the set $V$ is partitioned into three subsets:
 \begin{itemize}
  \item The {\em opinion leaders}: 	\tabto{26ex} $\tL(G)=\{i\in V\mid P_G(i)=\emptyset\text{ and } S_G(i)\neq\emptyset\}$.
  \item The {\em followers}: 		\tabto{26ex} $\tF(G)=\{i\in V\mid P_G(i)\neq\emptyset\text{ and } S_G(i)=\emptyset\}$.
  \item The {\em independent actors}: 	\tabto{26ex} $\tI(G)=\{i\in V\mid P_G(i)=\emptyset\text{ and } S_G(i)=\emptyset\}$.
 \end{itemize}
\end{definition}

When there is no risk of ambiguity, we simply use $S(i)$, $P(i)$, $\tI$, $\tL$ or $\tF$ instead of $S_G(i)$, $P_G(i)$, $\tI(G)$, $\tL(G)$ or $\tF(G)$.
Note that in an OLF system $\cS=(V,E)$, if $(i,j)\in E$ then $i\in\tL$ and $j\in\tF$.

Now we define the collective decision process of an OLF system, according to the unanimity restriction considered by~\cite{BRS12}.

\begin{definition}\label{def:OLF_C}
Given an OLF system $\cS$, the {\em collective decision vector} $c=c_\cS(x)$ associated to an initial decision vector $x$ is defined as 
\begin{equation}\label{eq:OLF_c}
 c_i=\begin{cases}
 b   &  \mbox{if } x_j=b \mbox{ for all } j\in P_G(i),\\
 x_i & \text{otherwise}
\end{cases}
\end{equation}
such that followers maintain their own initial decisions unless all their opinion leaders share an opposite inclination. Thus, restrincting our analysis to situations in which $n$ is an odd number, the {\em collective decision function} $C_\cS(x)$ is defined as 
\begin{equation}\label{eq:OLF_C}
C_\cS(x) = 
\begin{cases}
 1 & \mbox{if } |\{i\in V\mid c_i=1\}| > |\{i\in V\mid c_i=0\}|,\\
 0 & \mbox{if } |\{i\in V\mid c_i=1\}| < |\{i\in V\mid c_i=0\}|.\\
\end{cases}
\end{equation}
corresponding to the alternative with the greatest number of ``votes'' in the final choice vector.
\end{definition}

Observe that an OLF system requires the number of actors to be odd in order to ensure that decisions by the simple majority rule can be reached~\cite{BRS11,BRS12}. That is why inequalities in expression~(\ref{eq:OLF_C}) are strict. If we want to consider an even number of actors, we could just replace either $>$ by $\geq$ or $<$ by $\leq$ in expression (\ref{eq:OLF_C}), in which case the results of the paper remain.

Furthermore, both leaders and independent actors always follow their own inclinations in the collective choice decision vector. A follower follows the unanimous decision among its predecessors or their own inclination.

\begin{example}\label{ex:OLF-example}
Figure \ref{fig:OLF-example} illustrates a bipartite digraph $G=(V,E)$ corresponding to an OLF system over a set of five actors.
For both initial decision vectors $x=(0,1,1,1,0)$ and $y=(1,1,1,1,0)$ we obtain the same collective decision vector $c_\cS(x)=c_\cS(y)=(1,1,1,1,0)$ and the same collective decision $C_\cS(x)=C_\cS(y)=1$.
\begin{figure}[t]
\centering
\begin{tikzpicture}[every node/.style={circle,scale=0.9}, >=latex]
\node[draw](a) at (1.5,0.0)[label=right:$1$] {};
\node[draw](b) at (0.0,1.0)[label=above:$2$] {};
\node[draw](c) at (1.0,1.0)[label=above:$3$] {};
\node[draw](d) at (2.0,1.0)[label=above:$4$] {};
\node[draw](e) at (3.0,1.0)[label=above:$5$] {};
\draw[->] (b) to node {}(a);
\draw[->] (c) to node {}(a);
\draw[->] (d) to node {}(a);
\end{tikzpicture}
\caption{Example of an opinion leader-follower system.\label{fig:OLF-example}}
\end{figure}
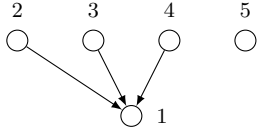
\end{example}

We finish this section with a preliminary result regarding the collective decision function of OLF systems. Recall from Section~\ref{sec:SAT} that the satisfaction score corresponds to the Rae index only for collective decision making models that are monotonic.

\begin{lemma}
\label{lem:monotony}
Let $\cS$ be an OLF system represented by a graph $G$, its corresponding collective decision function is monotonic, with respect to inclusion, on $\cP(V(G))$.
\end{lemma}

\begin{proof}
Let be $x\in\{0,1\}^n$. If $i\in\tL\cup\tI$, as $c_i(x)=x_i$ and $x_i\in\{0,1\}$, it is clear that $C(x-i)\leq C(x)\leq C(x+i)$.
If $i\in\tF$, we have three possibilities:\\
1) $x_i=0$ and $c_i(x)=1$, which implies $c_i(x+i)=1$, so $C(x)\leq C(x+i)$;\\
2) $x_i=1$ and $c_i(x)=0$, which implies $c_i(x-i)=0$, so $C(x-i)\leq C(x)$; and\\
3) $x_i=c_i(x)$, which is the same case as for opinion leaders and independent actors. \qed
\end{proof}

\subsection{Properties and axiomatization for \SAT in OLF systems}

For what follows, we denote a score as a function $f:V\to\matR$ that assigns some real value to each actor of the system. The following properties were introduced by~\cite{BRS11,BRS12}.

\begin{definition}\label{def:props_RAE_OLF}
Let $\cS$ and $\cS'$ be two OLF systems represented by the graphs $G$ and $G'$, respectively, such that $V(G)=V(G')$. Let $i,j,h$ be three different actors. We say that a measure given by the function $f:V\to\matR$ satisfies the properties:
 \begin{enumerate}
  \item {\bf Symmetry}:
	if $S(i)=S(j)$ and $P(i)=P(j)$, then $f(i)=f(j)$.
  \item {\bf Dictator property}:
	if $S(i)=V\bac\{i\}$, then $f(i)=2^n$.
  \item {\bf Dictated independence}:
	if $|P_G(i)|=|P_{G'}(i)|=1$, then $f_\cS(i)=f_{\cS'}(i)$.
  \item {\bf Equal gain property}:
	if $i\in\tL\cup\tI$, $j\in\tF$ and $E(G')=E(G)\cup\{(i,j)\}$,\\
	then $f_{\cS'}(i)-f_\cS(i)=f_{\cS'}(j)-f_\cS(j)$.
  \item {\bf Opposite gain property}:
	if $i\in\tL\cup\tI$, $j\in\tI$ and $E(G')=E(G)\cup\{(i,j)\}$,\\
	then $f_{\cS'}(i)-f_\cS(i)=f_\cS(j)-f_{\cS'}(j)$.
  \item {\bf Horizontal neutrality}:
	if $i\in\tL\cup\tI$, $j\in\tF$, $h\in\tL$, $E(G')=E(G)\cup\{(i,j)\}$ and $h\in P_G(j)$,\\
	then $f_{\cS'}(i)-f_\cS(i)=f_\cS(h)-f_{\cS'}(h)$.
 \end{enumerate}
\end{definition}

The above are desirable properties for scores. The symmetry property means that the score for actors with a symmetric position in the system is the same. A non-symmetrical measure could lead to unconventional results, e.g., two independent actors with different scores.

In this context, a {\em dictator} is an actor that points to all other actors of the system. Hence, in OLF systems there may be at most one dictator, and if $n>1$, the dictator is always an opinion leader. Furthermore, if there is a dictator, then all other actors follow this actor, so they adopt as final decision the initial decision of the dictator. The dictator property states that the dictators have the highest score as possible. Observe that this notion corresponds to the dictator player of simple games introduced in Definition~\ref{def:dummy_dictator}. Furthermore, this property is closely related to Lemma~\ref{lem:Rae_dummy_dictator}.

The dictated independence states that all the followers with only one opinion leader have the same score. However, note that a follower who has only one opinion leader has always to follow this opinion leader. Therefore, since any actor with only one predecessor is a dummy, then for Lemma~\ref{lem:Rae_dummy_dictator} the dictated independece is equivalent to the following:
$$\mbox{if } |P(i)|=1,\mbox{ then }f(i)=2^{n-1}.$$

The remaining properties involve changes in the structure of the OLF systems, by assigning to an actor a new opinion leader. These properties were inspired by similar properties for solution concepts in cooperative game theory~\cite{BRS11,Bri10}. One of the most relevant criteria to define a solution concept for cooperative games is the {\em fairness}, i.e., how well each player's payoff reflects its contribution~\cite{EP09,CEW11}. The most common solution concepts based on the fairness criterion are the power indices, and as we mention in Section~\ref{sec:SAT}, the satisfaction score is indeed equivalent to the Rae index, and it is closely related with the Banzhaf value. In particular, the equal gain property is closely related with the fairness concept by~\cite{Mye77}.

In a reasonable score, the addition of an influence relationship ---a directed edge--- from one actor to another should increase the score of the first actor, because now it is exerting more influence in the system. In this scenario, we can consider two cases:

On the one hand, if the influenced actor was a follower before the addition of the edge, then the score of this follower should also increase, because now it is more difficult to change its initial decision. The equal gain property states that when a follower gets an additional opinion leader, the changes in scores of this follower and of its new opinion leader are the same. For a score that does not meet this property, the addition of a relationship between these kind of actors could be unfair for one of them.

On the other hand, if the influenced actor was an independent actor, then the score of this actor should decrease, because its final decision now depends of the initial decision of the opinion leader. The opposite gain property states that when an independent actor gets an opinion leader, the sum of the scores of these two actors does not change. For a score that does not meet this property, the addition of a relationship between two actors could be unfair for the opinion leader, because it is not getting a profit according to the effort it took to influence the independent actor.

Finally, horizontal neutrality is inspired by the properties considered for collusion of players in cooperative games with transferable utility~\cite{Leh88,Hal94,Bri10}. This property states that, if a follower with at least one opinion leader gets an additional opinion leader, then the sum of scores of the old and new opinion leaders does not change. This means that the increase in the score for the new opinion leader comes fully from a decrease in the score for the other opinion leaders. For a score that does not meet this property, the new opinion leader could not get a profit according to the effort it took to influence an additional follower.\\

It is known that these properties hold for \SAT in OLF systems.

\begin{theorem}[\cite{BRS11,BRS12}]
\label{the:Rae_satisfy_OLF}
For OLF systems, the \SAT score satisfies the six properties of Definition~\ref{def:props_RAE_OLF}.
\end{theorem}

To show the axiomatization of satisfaction in OLF systems,~\cite{BRS12} introduced an additional axiom, which corresponds to the total sum of the satisfaction scores over all actors, i.e., a satisfaction normalization.

\begin{definition}\label{def:normalization_RAE_OLF}
Let $\cS$ be an OLF system represented by a graph $G=(V,E)$. A score given by the function $f:V\to\matR$ is normalized if it satisfies the following property:
 \begin{enumerate}
  \item[$7.$] {\bf Satisfaction normalization}:\\
	$\sum_{i\in V}f(i)=\sum_{x\in\{0,1\}^n}|\{i\in V\mid C(x)=x_i\}|$.
 \end{enumerate}
\end{definition}

Thus, these seven properties provide an axiomatization of the satisfaction for OLF systems.

\begin{theorem}[\cite{BRS12}]
\label{the:OLF_Rae_axioms}
For OLF systems, the \SAT score is the unique measure that satisfies the properties $1$, $2$, $3$, $4$, $5$, $6$ and $7$ of Definitions~\ref{def:props_RAE_OLF} and~\ref{def:normalization_RAE_OLF}.
\end{theorem}

\section{Generalized model: OLFM systems}
\label{sec:OLFM}

In this section we propose the {\em opinion leader-follower through mediators systems} ---OLFM systems--- as a generalization of the OLF systems.
OLFM systems allow to model decision making situations with mediators, i.e., actors that behave as opinion leaders and followers, in the sense that they receive their influence from opinion leaders or other mediators, and can influence the followers or other mediators.

Therefore, while OLF systems are supported on directed bipartite graphs, OLFM systems are supported on layered digraphs. 

\begin{definition}
A {\em layered digraph} is a digraph $G=(V,E)$ where $V$ can be partitioned into $k$ subsets $\cL_1,\ldots,\cL_k$ called {\em layers}, so that every edge connects a vertex from one layer to another vertex in a layer immediately below, i.e., for all $(a,b)\in E$, $a\in\cL_i$ and $b\in\cL_{i+1}$, for some $1\leq i<k$.
\end{definition}

This generalization allows to represent more complex social structures in which there are more than only two hierarchical levels.

\begin{definition}\label{def:OLFM}
An {\em opinion leader-follower through mediators system} $\cS$ ---an {\em OLFM system}, in short--- for a set of $n$ actors is given by a layered digraph $G=(V,E)$, such that the set $V$ is partitioned into four subsets:
 \begin{itemize}
  \item The opinion leaders: 	\tabto{26ex} $\tL(G)=\{i\in V\mid P_G(i)=\emptyset\text{ and } S_G(i)\neq\emptyset\}$.
  \item The followers: 		\tabto{26ex} $\tF(G)=\{i\in V\mid P_G(i)\neq\emptyset\text{ and } S_G(i)=\emptyset\}$.
  \item The independent actors: \tabto{26ex} $\tI(G)=\{i\in V\mid P_G(i)=\emptyset\text{ and } S_G(i)=\emptyset\}$.
  \item The {\em mediators}: 	\tabto{26ex} $\tM(G)=\{i\in V\mid P_G(i)\neq\emptyset\text{ and } S_G(i)\neq\emptyset\}$.
 \end{itemize}
\end{definition}

As for OLF systems, for OLFM systems we also restrict our attention to an odd number of actors. Both the collective decision vector and the collective decision function of the system is the same than for OLF systems ---see expressions (\ref{eq:OLF_c}) and (\ref{eq:OLF_C}) in Definition~\ref{def:OLF_C}---. However, here the collective decision vector must be determined in order, starting from the actors in the first layer, then the ones in the second layer, and so on.

Observe that the opinion leaders and independent actors belong to the first layer of the graph, $\cL_1$. The mediators are distributed into {\em layers of mediation}, whereas there are no mediators pointing to upper layers. The opinion leaders can only be connected with the mediators of the first layer of mediation, $\cL_2$; the mediators of the last layer of mediation can only be connected with the followers, and the mediators of interlayers can only be connected with the mediators of the layer immediately below. Moreover, $\cL_1=\tL\cup\tI$ and for all $i\in\cL_k$, $i\in\tF$. Hence, the OLF systems can be seen as OLFM systems with only two layers, i.e., with $k=2$.

Note also that the influence of actors in higher layers can affect the actors' decision in much lower layers. From Lemma~\ref{lem:monotony}, it is easy to see that the collective decision function under OLFM systems ---like in the OLF systems--- is monotonic.

\begin{example}\label{ex:OLFM_three_layer}
Figure \ref{fig:OLFM_three_layer} illustrates a graph $G$ corresponding to an OLFM system over a set of seven actors.
Here $\tL=\{1,2\}$, $\tI=\{3\}$, $\tM=\{4,5\}$ and $\tF=\{6,7\}$.
The computation of the collective decision function is shown in Table~\ref{tab:ex_c(x)_OLFM}, where the initial decision vectors are ordered according to binary numeration.

\begin{figure}[t]
\centering
\begin{tikzpicture}[every node/.style={circle,scale=0.9}, >=latex]
\node[draw](a) at (1,2.0)[label=above:$1$] {};
\node[draw](b) at (3,2.0)[label=above:$2$] {};
\node[draw](c) at (5,2.0)[label=above:$3$] {};
\node[draw](d) at (0,1.0)[label=left :$4$] {};
\node[draw](e) at (2,1.0)[label=right:$5$] {};
\node[draw](f) at (4,1.0)[label=right:$6$] {};
\node[draw](g) at (1,0.0)[label=right:$7$] {};
\draw[->] (a) to node {}(d);
\draw[->] (a) to node {}(e);
\draw[->] (b) to node {}(e);
\draw[->] (b) to node {}(f);
\draw[->] (d) to node {}(g);
\draw[->] (e) to node {}(g);
\end{tikzpicture}
\caption{An OLFM system with one layer of mediation.\label{fig:OLFM_three_layer}}
\end{figure}
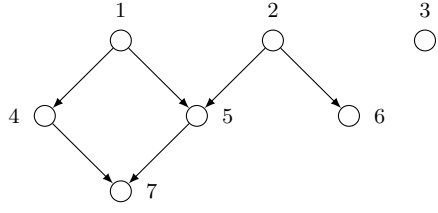

The vertical suspension points on the table indicate that both the collective decision vector and the collective decision function are the same than for the previous and the next decision vector.
Thus, according to our computation we have $\SAT(1)=104$, $\SAT(2)=\SAT(5)=88$, $\SAT(3)=\SAT(7)=72$ and $\SAT(4)=\SAT(6)=64$.

\begin{table}[t]
\centering
{\small
\begin{tabular}{|c|c|@{\,\,}c@{\,\,}!{\vrule width 1pt}c|c|@{\,\,}c@{\,\,}!{\vrule width 1pt}c|c|@{\,\,}c@{\,\,}|}\hline
 $x$     & $c(x)$  & $C(x)$ & $x$ & $c(x)$ & $C(x)$ & $x$ & $c(x)$ & $C(x)$\\\hline
 0000000 &         &   & 0110100 & 0110110 & 1 & 1001110 & 1001101 & 1 \\[-1.5ex]
 $\vdots$& 0000000 & 0 & 0110101 & 0110111 & 1 & 1001111 & 1001101 & 1 \\
 0001111 &         &   & 0110110 & 0110110 & 1 & 1010000 & 1011000 & 0 \\
 0010000 &         &   & 0110111 & 0110111 & 1 & 1010001 & 1011001 & 1 \\[-1.5ex]
 $\vdots$& 0010000 & 0 & 0111000 & 0110010 & 0 & 1010010 & 1011000 & 0 \\
 0011111 &         &   & 0111001 & 0110010 & 0 & 1010011 & 1011001 & 1 \\
 0100000 & 0100010 & 0 & 0111010 & 0110010 & 0 & 1010100 & 1011101 & 1 \\
 0100001 & 0100010 & 0 & 0111011 & 0110010 & 0 & 1010101 & 1011101 & 1 \\
 0100010 & 0100010 & 0 & 0111100 & 0110110 & 1 & 1010110 & 1011101 & 1 \\
 0100011 & 0100010 & 0 & 0111101 & 0110111 & 1 & 1010111 & 1011101 & 1 \\
 0100100 & 0100110 & 0 & 0111110 & 0110110 & 1 & 1011000 & 1011000 & 0 \\
 0100101 & 0100111 & 1 & 0111111 & 0110111 & 1 & 1011001 & 1011001 & 1 \\
 0100110 & 0100110 & 0 & 1000000 & 1001000 & 0 & 1011010 & 1011000 & 0 \\
 0100111 & 0100111 & 1 & 1000001 & 1001001 & 0 & 1011011 & 1011001 & 1 \\
 0101000 & 0100000 & 0 & 1000010 & 1001000 & 0 & 1011100 & 1011101 & 1 \\
 0101001 & 0100000 & 0 & 1000011 & 1001001 & 0 & 1011101 & 1011101 & 1 \\
 0101010 & 0100010 & 0 & 1000100 & 1001101 & 1 & 1011110 & 1011101 & 1 \\
 0101011 & 0100010 & 0 & 1000101 & 1001101 & 1 & 1011111 & 1011101 & 1 \\
 0101100 & 0100110 & 0 & 1000110 & 1001101 & 1 & 1100000 &         &   \\[-1.5ex]
 0101101 & 0100111 & 1 & 1000111 & 1001101 & 1 & $\vdots$& 1101111 & 1 \\
 0101110 & 0100110 & 0 & 1001000 & 1001000 & 0 & 1101111 &         &   \\
 0101111 & 0100111 & 1 & 1001001 & 1001001 & 0 & 1110000 &         &   \\[-1.5ex]
 0110000 & 0110010 & 0 & 1001010 & 1001000 & 0 & $\vdots$& 1111111 & 1 \\
 0110001 & 0110010 & 0 & 1001011 & 1001001 & 0 & 1111111 &         &   \\
 0110010 & 0110010 & 0 & 1001100 & 1001101 & 1 &         &         &   \\
 0110011 & 0110010 & 0 & 1001101 & 1001101 & 1 &         &         &   \\\hline
\end{tabular}
}
\caption{Collective decision function for an OLFM system.\label{tab:ex_c(x)_OLFM}}
\end{table}
\end{example}

It is interesting to note that in OLFM systems, instead of OLF systems, the actors of the same kind may not have the same satisfaction score. Observe that in Example~\ref{ex:OLFM_three_layer}, for instance, the satisfaction of a follower may be greater than the satisfaction of a mediator, and equal than the satisfaction of an independent actor.

\subsection{Properties and axiomatization for \SAT in OLFM systems}

In this subsection we shall prove first that all the properties for satisfaction in OLF systems also apply for OLFM systems. However, to establish an axiomatization in OLFM systems, we need to generalize the equal gain property and the opposite gain property, in order to consider the mediators in the layered graphs. Although it is not required for the axiomatization, we also introduce a generalization of the horizontal neutrality that is fulfilled for satisfaction in OLFM systems.

\begin{definition}\label{def:props_RAE_OLFM}
Let $\cS$ and $\cS'$ be two OLFM systems represented by the graphs $G$ and $G'$, respectively, such that $V(G)=V(G')$. Let $i,j,h$ be three different actors, and $k\in\matN$ such that $k\geq 0$. We say that a measure given by the function $f:V\to\matR$ satisfies the properties:
 \begin{enumerate}
  \item[$4b$] {\bf Equal absolute change property}:\\
	if $i\in\cL_{k-1}$, $j\in\cL_k$ and $E(G')=E(G)\cup\{(i,j)\}$,\\
	then either
	$f_{\cS'}(i)-f_\cS(i)=f_{\cS'}(j)-f_\cS(j)$ or
	$f_{\cS'}(i)-f_\cS(i)=f_\cS(j)-f_{\cS'}(j)$.
  \item[$5b$] {\bf Opposite gain property}:\\
	if $i\in{}V$, $j\in\tI$ and $E(G')=E(G)\cup\{(i,j)\}$,\\
	then either
	$f_{\cS'}(i)-f_\cS(i)=f_\cS(j)-f_{\cS'}(j)$ or
	$f_{\cS'}(i)-f_\cS(i)=f_{\cS'}(j)-f_\cS(j)$.
  \item[$6b$] {\bf Power neutrality for two opinion leaders}:\\
	if $h\in\cL_{k-1}$, $i\in\cL_{k-1}$, $j\in\cL_k$ with $P_G(j)=\{h\}$, and furthermore, $E(G')=E(G)\cup\{(i,j)\}$,\\
	then either
	$f_{\cS'}(i)-f_\cS(i)=f_\cS(h)-f_{\cS'}(h)$ or
	$f_{\cS'}(i)-f_\cS(i)=f_{\cS'}(h)-f_\cS(h)$.
 \end{enumerate}
\end{definition}

Note that the opposite gain property is a generalization of the property 5 of Definition~\ref{def:props_RAE_OLF}, because when $i\in\tL\cup\tI$, it only holds $f_{\cS'}(i)-f_\cS(i)=f_\cS(j)-f_{\cS'}(j)$. The equal absolute change property is a generalization of equal gain property, because when $i\in\tL\cup\tI$, it only holds $f_{\cS'}(i)-f_\cS(i)=f_{\cS'}(j)-f_\cS(j)$. The power neutrality for two opinion leaders is a generalization of horizontal neutrality, because for $k=2$, it only holds $f_{\cS'}(i)-f_\cS(i)=f_\cS(h)-f_{\cS'}(h)$.
Moreover, the properties $4b$ and $6b$ were introduced by~\cite{BRS11} for OLF systems ---i.e., OLFM systems with two layers--- not restricted to unanimity, i.e., so that followers can change their decisions based on a majority proportion of their opinion leaders.

The following result proves that all the previous properties are fulfilled by the satisfaction in OLFM systems.

\begin{theorem}
\label{the:OLFM_Rae_props}
For OLFM systems, the \SAT score satisfies the properties $1$, $2$, $3$, $4b$, $5b$ and $6b$ of Definitions~\ref{def:props_RAE_OLF} and~\ref{def:props_RAE_OLFM}.
\end{theorem}

\begin{proof}
For symmetry, for all $x\in\{0,1\}^n$, $P(i)=P(j)$ implies $c_i(x)=c_j(x)$. Further, as $S(i)=S(j)$, if $x_i\neq x_j$, then $c(x)=c(x-i+j)$; and if $x_i=x_j$, the satisfaction score does not change for the actors $i$ and $j$.

For the dictator property, if $S(i)=V\bac\{i\}$ we have that $i\in\tL$ and $|\tI|=|\tM|=0$,
which is the same case proved for OLF systems in~\cite{BRS11}, i.e., as $C(x)=x_i$ for all $x\in\{0,1\}^n$, then $\SAT(i)=2^n$.

For the dictated independence, let be $P(i)=\{j\}$, then for all $x\in\{0,1\}^n$ it holds $c_i(x)=x_j$, so the collective choice $C(x)$ is independent of the decision of the actor $i$. Hence, let be $b=\{0,1\}$, if $C(x)=b$, there are exactly $2^{n-1}$ initial decision vectors with $x_i=b$, and $2^{n-1}$ with $x_i=1-b$.

For what follows, note that for every $x\in\{0,1\}^n$ such that $C_\cS(x)=C_{\cS'}(x)$, it holds $\SATbar_\cS(i,x)=\SATbar_{\cS'}(i,x)$, for all $i\in{}V$.
Therefore, to determine $\SAT_\cS(i)$ and $\SAT_{\cS'}(i)$ we only need to consider the initial decision vectors $x\in\{0,1\}^n$ where $C_\cS(x)\neq{}C_{\cS'}(x)$.

For the equal absolute change property, first consider that $i\in\tL\cup\tI$ and $j\in\tM\cup\tF$.
As $c_{j,\cS}(x)\neq{}x_i$ and $c_{j,\cS'}(x)=x_i$, then $C_\cS(x)\neq{}x_i$ and $C_{\cS'}(x)=x_i$; hence $\SATbar_{\cS'}(i,x)-\SATbar_\cS(i,x)=1$.
If $c_{j,\cS}(x)=x_j$, then $x_j\neq{}x_i$, so $C_\cS(x)=x_j$ and $C_{\cS'}(x)\neq{}x_j$, which implies $\SATbar_\cS(j,x)-\SATbar_{\cS'}(j,x)=1$; and
if $c_{j,\cS}(x)\neq{}x_j$, then $x_j=x_i$, so $C_\cS(x)\neq{}x_j$ and $C_{\cS'}(x)=x_j$, implying $\SATbar_{\cS'}(j,x)-\SATbar_\cS(j,x)=1$.
The possible change of inclinations or decisions of successors of $j$ keeps that $C_\cS(x)\neq{}C_{\cS'}(x)$, and this does not contradicts the above.
Thus, by expression~(\ref{SATwithSATbar}) we have either $\SAT_{\cS'}(i)-\SAT_\cS(i)=\SAT_{\cS'}(j)-\SAT_\cS(j)$ 
or $\SAT_{\cS'}(i)-\SAT_\cS(i)=\SAT_\cS(j)-\SAT_{\cS'}(j)$.

Second, consider $i\in\tM$. Note that in this case, the inclination of actor $i$ also depends of their predecessors.
To deal with this, just replace $x_i$ in all the above equations by $c_i(x)$, and note that $c_i(x)=c_{i,\cS}(x)=c_{i,\cS'}(x)$,
so if $x_i=c_i(x)$, we obtain the same equations, and if $x_i\neq{}c_i(x)$, we obtain that $\SATbar_\cS(i,x)-\SATbar_{\cS'}(i,x)=1$,
getting the same final equations that above.

For the opposite gain property, first consider $i\in\tL\cup\tI$ and $j\in\tI$.
As it must hold that $x_i\neq{}c_j$, then $C_\cS(x)=x_j\neq{}x_i$ and $C_{\cS'}(x)=x_i\neq{}x_j$;
hence $\SATbar_{\cS'}(i,x)-\SATbar_\cS(i,x)=1$ and $\SATbar_\cS(j,x)-\SATbar_{\cS'}(j,x)=1$.
Second, consider that $i\in\tM\cup\tF$. For this case, just replace $x_i$ in all the above equations by $c_i(x)$, and note that $c_i(x)=c_{i,\cS}(x)=c_{i,\cS'}(x)$,
so it holds $c_i(x)\neq x_j$, $C_\cS(x)=x_j$ and $C_{\cS'}(x)=c_i(x)$;
hence, $\SATbar_\cS(j,x)-\SATbar_{\cS'}(j,x)=1$ and either $\SATbar_{\cS'}(i,x)-\SATbar_\cS(i,x)=1$ or $\SATbar_\cS(i,x)-\SATbar_{\cS'}(i,x)=1$.
Thus, by expression~(\ref{SATwithSATbar}) we have either
$\SAT_{\cS'}(i)-\SAT_\cS(i)=\SAT_\cS(j)-\SAT_{\cS'}(j)$ or $\SAT_{\cS'}(i)-\SAT_\cS(i)=\SAT_{\cS'}(j)-\SAT_\cS(j)$.

For the power neutrality for two opinion leaders, first consider $i\in\tL\cup\tI$, $j\in\tM\cup\tF$ and $h\in\tL$.
As $|P_G(j)|=1$, $c_{j,\cS}(x)=x_h$, and as $|P_{G'}(j)|=2$, $c_{j,\cS'}(x)\neq{}x_j$ iff $x_h=x_i\neq{}x_j$.
Let $b\in\{0,1\}$, if $C_\cS(x)=b$ and $C_{\cS'}(x)=1-b$, then $c_{j,\cS}(x)=b=x_h$ and $c_{j,\cS'}(x)=1-b=x_i=x_j$,
hence $\SATbar_{\cS'}(i,x)-\SATbar_\cS(i,x)=1=\SATbar_\cS(h,x)-\SATbar_{\cS'}(h,x)$.
The possible change of inclinations of successors of $j$ keeps that $C_\cS(x)\neq{}C_{\cS'}(x)$, and this does not contradicts the above.
Thus, by expression~(\ref{SATwithSATbar}) we have $\SAT_{\cS'}(i)-\SAT_\cS(i)=\SAT_\cS(h)-\SAT_{\cS'}(h)$.

Second, consider $h\in\tM$. For this case, just replace $x_h$ in all the above equations by $c_h(x)$, and note that $c_h(x)=c_{h,\cS}(x)=c_{h,\cS'}(x)$, so either $\SATbar_\cS(h,x)-\SATbar_{\cS'}(h,x)=1$ or $\SATbar_{\cS'}(h,x)-\SATbar_\cS(h,x)=1$.
Finally, consider $i\in\tM\cup\tF$. Replacing $x_i$ by $c_i(x)$, where $c_i(x)=c_{i,\cS}(x)=c_{i,\cS'}(x)$, we obtain analogous equations.
Therefore we have either $\SAT_{\cS'}(i)-\SAT_\cS(i)=\SAT_\cS(h)-\SAT_{\cS'}(h)$ or $\SAT_{\cS'}(i)-\SAT_\cS(i)=\SAT_{\cS'}(h)-\SAT_\cS(h)$. \qed
\end{proof}

Note that the satisfaction normalization of Definition~\ref{def:normalization_RAE_OLF} remains the same for OLFM systems, because the collective decision function is the same. From the previous theorem, since for OLFM systems the satisfaction score satisfies power neutrality for two opinion leaders, then it also satisfies horizontal neutrality. In what follows we prove an axiomatization of satisfaction for OLFM systems.

\begin{theorem}
For OLFM systems, the \SAT score is the unique measure that satisfies the properties $1$, $2$, $3$, $4b$, $5b$, $6$ and $7$ of Definitions~\ref{def:props_RAE_OLF}, \ref{def:normalization_RAE_OLF} and~\ref{def:props_RAE_OLFM}.
\end{theorem}

\begin{proof}
We know by Theorem~\ref{the:OLFM_Rae_props} that in OLFM systems the \SAT score satisfies properties $1$, $2$, $3$, $4b$, $5b$, $6$ and $7$. To prove uniqueness it remains to show that, on the assumption that $f:V\to\matR$ satisfies the seven axioms, then this score must be equal to \SAT.

By Theorem~\ref{the:OLF_Rae_axioms}, we know that if there are no mediators ---i.e., we have an OLF system---, property $4b$ is replaced by property $4$, so the score is equal to \SAT.
Now we proceed constructively.

First, given an OLFM system $\cS$ without mediators, we can transform a follower $i$ in a mediator by connecting it with an independent actor $j$, obtaining a new OLFM system $\cS'$.
Thus, by property~$5b$, it holds either $f_{\cS'}(i)-f_\cS(i)=f_\cS(j)-f_{\cS'}(j)$ or $f_{\cS'}(i)-f_\cS(i)=f_{\cS'}(j)-f_\cS(j)$.
As $f_\cS(i)$ and $f_\cS(j)$ are uniquely determined by Theorem~\ref{the:OLF_Rae_axioms}, both equations yield a system of linear equations easy to solve,
so that the unknowns, $f_{\cS'}(i)$ and $f_{\cS'}(j)$, can be uniquely determined.

From the above, note that actor $j$ in $\cS'$ becomes in a follower.
And also note that we can transform step by step other independent actors $j$ in followers, such that $P_{G'}(j)=\{i\}$.
In each step, satisfaction score can be uniquely determined by using the same property.

Secondly, suppose that we have an OLFM system $\cS$ with only one layer of mediation, like the obtained above, with a follower $j\in\cL_3$ so that $P_{G'}(j)=\{i\}$.
Now we can transform a follower $h\in\cL_2$ in a mediator, by connecting it with follower $j$, obtaining a new OLFM system $\cS'$.
Thus, by property~$4b$, it holds either $f_{\cS'}(h)-f_\cS(h)=f_\cS(j)-f_{\cS'}(j)$ or $f_{\cS'}(h)-f_\cS(h)=f_{\cS'}(j)-f_\cS(j)$.
This is basically the same kind of system of linear equations obtained with property~$5b)$, and as $f_\cS(h)$ and $f_\cS(j)$ are uniquely determined, then $f_{\cS'}(h)$ and $f_{\cS'}(j)$ can also be uniquely determined.
We can also repeat this process by transforming new followers $h\in\cL_2$ in mediators, obtaining in each step that satisfaction can be uniquely determined.

Of course, the same kind of transformations can be done to create lower layers, and therefore to produce any OLFM system.

Finally, note that property~$1$ implies that there is a constant $c\in\matR$ such that for all $i\in\tI$, $f(i)=c$.
Hence, for every OLFM system, we can provide new independent actors and then using them as opinion leaders, followers or mediators, in such a way that $f$ can always be uniquely determined. \qed
\end{proof}

\section{Concluding remarks}

We have studied the satisfaction score proposed by~\cite{BRS11} for a generalization of OLF systems, namely the OLFM systems. This generalized model incorporates an additional kind of actors called mediators, that act at the same time as opinion leaders and followers. Mediators allow the presence of several layers of influence, and hence they establish a more general hierarchy among the different actors.

By using this generalized model, the main results of this work are two: first, that the properties of satisfaction for OLF systems support natural generalizations for OLFM systems; and second, that these properties allow to establish a new axiomatization of the score over the generalized model, when the collective decision vectors are restricted by the expression~(\ref{eq:OLF_c}). This axiomatization extends the results of~\cite{BRS12} for OLF systems.

Interestingly enough, the measure coincides with the well established \RAE index, that is closely related to the Banzhaf value over the set of monotonic decision functions which can be casted as characteristic functions of simple games. Equation~(\ref{eq:RAE_BZ}) suggests that through small modifications, such as a change in the normalization of property $7$, it is possible to define an axiomatization of the Banzhaf value for OLFM systems. In the same vein, note that there exist other axiomatizations of the Banzhaf value for simple games~\cite{DS79}. Moreover, for OLF systems,~\cite{BRS11,BRS12} describe another score called {\em power score}, that is very similar to the satisfaction in terms of the properties that it meets. Analogously to \SAT and \RAE, the power score is closely related with the power index called {\em Holler value}~\cite{Hol82}, also studied in the context of simple game theory. It remains open to study whether the power score also meets the generalized versions of the properties for OLFM systems, as well as if it admits an axiomatization for this generalized model.

Additionally, it may be interesting to determine how far we can extend the model so that the properties are still met, and even how far it is possible to establish an axiomatization of the scores.
For instance, we may consider more general collective decision vectors, like the one described by~\cite{BRS11}:\footnote{In~\cite{BRS11} the authors define these collective decision vectors to study the satisfaction score, but they do not define an axiomatization for the satisfaction in the systems that use them. The known axiomatization that we mention in this paper is defined in~\cite{BRS12} for the collective decision vectors of Definition~\ref{def:OLF_C}.}
\[
 c_i=\begin{cases}
 b   & \text{if  $|\{j\in P_G(i)\mid x_j=b\}|>\lfloor q\cdot|P_G(i)|\rfloor,$}\\
 x_i & \text{otherwise}
\end{cases}
\]
where $b\in\{0,1\}$ and $1/2\leq q<1$. This parameter $q$, the {\em fraction value}, represents the fraction of opinion leaders with the same inclination that is necessary to influence the decision of a follower. Observe that the unanimity condition corresponds to consider the case when the fraction value $q$ is large enough. Note also that by Lemma~\ref{lem:monotony}, the system remains monotonic. From these generalized collective decision vectors emerge both the equal absolute change property and the power neutrality for two opinion leaders in OLF systems.

Another possibility is to consider connections among actors that belongs to not immediately consecutive layers. A first approach could be dealing with {\em star mediation influence games}, a collective decision making model proposed in~\cite{MRS13}. This model is based on star graphs, so that we have five kind of actors: a unique mediator that acts as central actor, a set of opinion leaders that point to the mediator, a set of followers that are influenced by the mediator, another set of actors that point to the mediator and at the same time are influenced by it, and a set of independent actors. The interesting fact of this model is that despite it breaks the layered structure of the graphs, it is still simple, because the fraction value only affects the mediator. Other extremal cases of influence games were introduced in~\cite{MRS12}.

Finally, as pointed out by~\cite{BRS11}, other related models that can be considered are games with $r$ alternatives~\cite{Bol86,Bol93,Bol00,Bol02} or decision-making situations with several levels of approval in the input and the output, rather than only one~\cite{Fre05,Fre05b,FZ03}.

\section*{Acknowledgements}
The author thanks to Maria Serna and Xavier Molinero for their comments in the preparation of the paper.

\bibliographystyle{plain}

\end{document}